\documentclass[a4paper,12pt,final]{article}
\usepackage{fullpage}
\usepackage{color}
\usepackage{bbm}
\usepackage{amssymb,amsmath,amssymb,amsthm}
\usepackage{pdfpages}
\usepackage{hyperref} 
\usepackage[notcite,notref]{showkeys}
\usepackage{bm}
\usepackage{todonotes}

\newcommand\nn{^{(n)}}

\newtheorem{theorem}{Theorem}

\newcommand\comment[1]{}
\renewcommand{\c}{u}
\newcommand\cc{u}
\renewcommand\color[1]{}
\renewcommand\marginpar[1]{}

\title{Some remarks on a viscous regularization\\ of the nonlinear diffusion equation}
\author{Giuseppe Tomassetti\footnote{Universita` di Roma Tor Vergata, Dipartimento di Ingegneria Civile ed Ingegneria Informatica. Via Politecnico 1, 00133 Roma, Italy. Email: \texttt{tomassetti@ing.uniroma2.it}}}

\begin{document}
\maketitle
\begin{abstract}
We illustrate an alternative derivation of the viscous regularization of the diffusion equation which was studied in [A.~Novick-Cohen and R.~L. Pego. {\em Trans. Amer. Math. Soc.}, 324:331--351]. We provide an alternative proof of existence of solutions, based on the Galerkin method and on compactness arguments. In addition, we propose a ``non-smooth'' variant of the viscous regularization which we believe may result in interesting hysteretic effects.
\end{abstract}

{\bf Keywords:} diffusion, backward-parabolic partial differential equations, viscosity. 

\section{Introduction}
The \emph{nonlinear diffusion equation}:
\begin{align}
&\dot {\c}=\Delta(f({\c}))\label{diffus}
\end{align}
is the standard mathematical model for species diffusion in a continuous medium. The derivation of \eqref{diffus} begins with the \emph{mass--balance equation}:
  \begin{equation}
    \label{eq:1}
    \dot {\c}+{\rm div}\mathbf h=0,
  \end{equation}
which relates the time derivative of \emph{species concentration} $u$ to the divergence of the \emph{flux of diffusant} $\mathbf h$. Then, according to  \emph{Fick's law}:
\begin{equation}
    \label{eq:2}
    \mathbf h=-\alpha\nabla\mu,
  \end{equation} 
the flux of diffusant is deemed proportional to the \emph{chemical-potential gradient} $\nabla\mu$ through a constant, positive \emph{mobility} $\alpha$. The combination of \eqref{eq:1} and \eqref{eq:2} yields
\begin{equation}\label{eq:109}
  \dot u=\alpha\Delta\mu,
\end{equation}
and, for $\widehat\psi'(\cdot)$ the derivative of the \emph{coarse-grain free-energy mapping} $\widehat\psi(\cdot)$, the \emph{equation of state}:
\begin{equation}\label{eq:5}
  \mu=\widehat\psi'({\c})
\end{equation}
is enforced, so that \eqref{diffus} is arrived at on setting:
\begin{equation}
  f(u)=\alpha\widehat\psi'(u).
\end{equation}

When modeling diffusion coupled with phase separation, one usually selects a non-convex coarse-grain free energy, the typical example being the double-well potential:
\begin{equation}
  \widehat\psi(u)=\kappa u^2(u-1)^2,
\end{equation}
with $\kappa$ a positive constant. In this case, \eqref{diffus} is ill posed, and a regularization of the equation of state is in order, the most popular choice being the \emph{elliptic regularization}:
\begin{equation}\label{eq:101}
  \mu=\widehat \psi'({\c})-\lambda\Delta{\c},
\end{equation}
with $\lambda$ a positive constant related to surface tension, leading to the celebrated \emph{Cahn-Hilliard equation} \cite{CahnH1958JCP}:
\begin{align}\label{eq:108}
  &\dot u=\alpha\Delta(\widehat\psi'(u)-\lambda\Delta\mu).
\end{align}

An alternative approach was explored in \cite{NovicP1991TAMS}, where \eqref{eq:101} was replaced with
\begin{equation}
  \label{eq:53}
  \mu=\widehat\psi'({\c})+\beta\dot{\c},
\end{equation}
with $\beta$ a positive constant, whence the following \emph{viscous regularization of the diffusion equation}:
\begin{equation}\label{eq:13}
\dot{\c}-\alpha\beta\Delta\dot{\c}=\alpha\Delta\psi'({\c}),
\end{equation}
which was shown in the same paper to be well posed when $\widehat\psi'(\cdot)$ is a locally-Lipshitz function fulfilling certain growth assumptions. 

Besides \cite{NovicP1991TAMS}, there is a substantial amount of mathematical literature devoted to the analysis of \eqref{eq:13} and of similar models. The case when $\widehat\psi'(\cdot)$ is decreasing for large values of the argument has been considered in \cite{victor1998sobolev}. Related models have been studied in  \cite{barenblatt1993degenerate} and \cite{bertsch2013pseudoparabolic}. 

The papers \cite{EvansP2004MMMAS} and  \cite{Plotn1994Passing} are concerned with the behavior of the solutions in the vanishing-viscosity limit $\beta\to 0$ and provide insight on the hysteretic properties of the so--called entropy solutions. Such vanishing-viscosity approach provides a selection criterion for solutions outside the standard Sobolev setting in the degenerate case $\beta=0$, as shown in \cite{porzio2013radon} and \cite{thanh2014sobolev}. Forward-backward parabolic equations leading to hysteresis have also been considered in \cite{visintin2002forward}.

System obtained by combining both elliptic and viscous regularization have been considered  in \cite{EllioS1996JDE} and \cite{elliott1996cahn}, whose vanishing-viscosity limit has been studied in \cite{thanh2014passage}. More sophisticated generalizations of the Cahn--Hilliard system that still incorporate a viscous constribution have been proposed and investigated in \cite{miranville1999model,miranville2000some,miranville1998dynamical}. 

Despite the impressive amount of analytical literature concerning \eqref{eq:101}, only a few references are available concerning its justification and interpretation. In \cite{Novic1988viscous} it was shown that \eqref{eq:13} may be recovered as a suitable limit of  the equations describing the motion of a mixture of two fluids. The possibility of giving a rational position to \eqref{eq:101} is also intimated in \cite{FriedG1999JSP}. A derivation based on microforce balance is also provided in \cite{miranville1999model}.

Here, following the point of view of \cite{GurtiFA2010}, we select a quite encompassing constitutive class (\emph{cf.} \eqref{eq:111}) and, within this class, we find the most general constitutive equations consistent with the second law of thermodynamics (\emph{cf.} \eqref{eq:30}), appropriate to the isothermal context. In particular, we are able to retrieve \eqref{eq:53} as a special case (\emph{cf.} \eqref{eq:19}). In addition we propose a proof of existence of weak solutions for a certain class of polynomial coarse-grain free energies. One of our motivations for presenting this proof is to illustrate how the viscous regularization affects the relevant estimates. 

The last section of this paper collects a handful of additional issues: we propose a non--smooth variant of \eqref{eq:53} (\emph{cf.} \eqref{eq:35} and \eqref{eq:36}), and we argue that this variant should produce interesting hysteresis effects; then, we cursorily examine singular free energies and we discuss the position of prescriptions like \eqref{eq:53} within the non-standard thermodynamical setting put forth in \cite{Podio2006RM}.

\section{Standard background}\label{sec:stdback}
The point of view promoted by Gurtin in \cite{Gurti1996PD}, and illustrated in the recent monograph \cite{GurtiFA2010}, is that chemical potential must be treated as a \emph{primitive field} representing the amount of energy carried by a unit amount of diffusant. In accordance with this notion, given any part $\mathcal P$ of the body where the diffusion process takes place, the quantity
\begin{equation}
  \label{eq:6}
  \mathcal T(\mathcal P)=-\int_{\partial\mathcal P}\mu\mathbf h\color{blue}\,{\rm d}\Gamma
\end{equation}
must be interpreted as \emph{the amount of chemical energy supplied to $\mathcal P$  per unit time}. In the absence of mechanical interactions, the appropriate version of the dissipation inequality is: 
\begin{equation}
  \label{eq:7}
  \int_{\mathcal P}\dot\psi\color{blue}\,{\rm d}x\color{black}\le \mathcal T(\mathcal P).
\end{equation}
By combining \eqref{eq:6} and \eqref{eq:7}, using the divergence theorem, and a localization argument, one obtains the inequality $\dot\psi+\mu\,\textrm{div}\mathbf h+\mathbf h\cdot\nabla\mu\le 0$; then with the help of the mass-balance equation \eqref{diffus}, one easily arrives at the \emph{local dissipation inequality}:
\begin{equation}
 \label{eq:9}
\dot\psi-\mu\dot{\c}+\mathbf h\cdot\nabla\mu\le 0.
\end{equation}
It is standard practice in continuum mechanics to exploit dissipation inequalities to
\begin{itemize}
\item single out those quantities that should be the object of constitutive specification: in the present case, the triplet $(\psi,\mu,\mathbf h)$;
\item determine what fields should appear as independent variables in these specifications: here, the triplet $({\c},\dot{\c},\mathbf g)$, with
  \begin{equation}
    \label{eq:14}
    \mathbf g=\nabla\mu.
  \end{equation}
\item discard, in the manner of Coleman and Noll \cite{ColemN1963ARMA}, ``unphysical'' constitutive choices, based on the requirement that the dissipation inequality  be never violated along any realizable process.
\end{itemize}
An example of application of this procedure may be found in the recent monograph \cite{GurtiFA2010}, whose Section 66 presents a fully-fledged constitutive theory for single-species transport coupled with elasticity. The starting point of the theory in question are certain  provisional constitutive equations which, when strain is neglected, take the form:
\begin{equation}
  \label{eq:11}
  \psi=\widehat\psi({\c}),\quad \mu=\widehat\mu({\c}),\quad \mathbf h=\widehat{\mathbf h}({\c},\mathbf g).
\end{equation}
It is shown in \cite{GurtiFA2010} that consistency  of \eqref{eq:11} with the dissipation inequality \eqref{eq:9} demands that the constitutive mappings delivering free energy and chemical potential be related by
\begin{equation}
  \label{eq:15}
\hat\mu({\c})= \hat\psi'({\c}),
\end{equation}
and that 
\begin{equation}
  \label{eq:24}
  \widehat{\mathbf h}({\c},\mathbf g)=-\widehat{\mathbf M}({\c},\mathbf g)\mathbf g
\end{equation}
with the \emph{tensorial-mobility mapping} $\widehat{\mathbf M}$ satisfying:
\begin{equation}\label{eq:22}
  \widehat{\mathbf M}({\c},\mathbf g)\mathbf g\cdot\mathbf g\ge 0.
\end{equation}
It is important to notice that:
\begin{itemize}
\item  the relation \eqref{eq:5} between chemical potential and free--energy mapping is recovered as a derived assertion;
\item  the (isotropic) Fick's law \eqref{eq:2} can be recovered as a special case of \eqref{eq:24}, by choosing:
\begin{equation}
  \label{eq:25}
  \widehat{\mathbf M}({\c},\mathbf g)=\alpha\mathbf I,
\end{equation}
where the constant $\alpha$ is non-negative, as demanded by \eqref{eq:22}.
\end{itemize}
For the sake of completeness, we briefly recapitulate the steps leading to \eqref{eq:15}--\eqref{eq:22}, and we refer to \cite{GurtiFA2010} for details. We begin by noticing that, on account of \eqref{eq:11}, the dissipation inequality \eqref{eq:9} becomes:
\begin{equation}
  \label{eq:12}
 (\widehat\psi'({\c})-\widehat\mu({\c}))\dot{\c}+\widehat{\mathbf h}({\c},\mathbf g)\cdot\mathbf g\le 0.
\end{equation}
Following Coleman \& Noll \cite{ColemN1963ARMA}, we consider a process such that, at a given point, and at a given time, $\mathbf g=\mathbf 0$, with ${\c}$ and $\dot{\c}$ attaining arbitrary values. At that particular point and time, the dissipation inequality specializes to:
\begin{equation}
  \label{eq:155}
  (\widehat\psi'({\c})-\widehat\mu({\c}))\dot{\c}\le 0.
\end{equation}
We insists on asking that \eqref{eq:155} be satisfied for whatever choice of ${\c}$ and $\dot{\c}$. Such requirement can be met only if the constitutive mappings for free energy and chemical potential are related by \eqref{eq:15}.

In view of \eqref{eq:15}, what is left with \eqref{eq:12} is the so--called \emph{residual inequality}:
\begin{equation}
  \label{eq:16}
  \widehat{\mathbf h}({\c},\mathbf g)\cdot\mathbf g\le 0.
\end{equation}
Further conclusions can be drawn from \eqref{eq:16} by fixing ${\c}$ and $\mathbf g$, and by looking at the function:
\begin{equation}
  \label{eq:31}
\lambda\mapsto \widehat{\mathbf h}({\c},\lambda\mathbf g)\cdot\mathbf g.
\end{equation}
By \eqref{eq:16}, the function specified in \eqref{eq:31} changes its sign at $\lambda=0$. Moreover, this function is smooth, granted that the constitutive mapping $\widehat{\mathbf h}$ is smooth. Then, we have necessarily $\widehat{\mathbf h}({\c},\mathbf 0)\cdot\mathbf g=0$. Since ${\c}$ and $\mathbf g$ can be chosen arbitrarily, we conclude that:
\begin{equation}
  \label{eq:34}
  \widehat{\mathbf h}({\c},\mathbf 0)=\mathbf 0.
\end{equation}
In words: if the gradient of chemical potential is null, then the flux of diffusant is null as well. As a consequence of \eqref{eq:34}, the constitutive mapping $\widehat{\mathbf h}$ admits the representation \eqref{eq:24}, with \eqref{eq:22} being required by the residual inequality \eqref{eq:16}.

\section{The parabolic regularization and its generalizations}\label{sec:parabreg}
Willing to explore constitutive dependencies more general than \eqref{eq:11}, we notice that the free-energy imbalance prompts the inclusion of $\dot{\c}$ in the set of independent variables appearing in the constitutive equations. We start from the following generalization of \eqref{eq:11}:
\begin{equation}
  \label{eq:111}
  \psi=\widetilde\psi({\c},\dot{\c}),\quad \mu=\widetilde\mu({\c},\dot{\c}),\quad \mathbf h=\widetilde{\mathbf h}({\c},\dot{\c},\mathbf g).
\end{equation}
We assume that the constitutive mappings appearing in \eqref{eq:111} are all \emph{smooth}. Then, the dissipation inequality \eqref{eq:9} takes the form:
\begin{equation}
  \label{eq:112}
(\partial_{\c}\widetilde\psi({\c},\dot{\c},\mathbf g)-\widetilde\mu({\c},\dot{\c}))\dot{\c}+\widetilde{\mathbf h}({\c},\dot{\c},\mathbf g)\cdot\mathbf g+\partial_{\dot{\c}}\widetilde\psi({\c},\dot{\c},\mathbf g) \ddot{\c}\le 0.
\end{equation}
If we assume that \eqref{eq:112} holds true for whatever continuation of a process, and hence for any arbitrary assignment of $\ddot{\c}$ and $\dot{\mathbf g}$, we see that $\partial_{\dot{\c}}\widetilde\psi=0$. Thus, there exists a function $\widehat\psi$ such that:
\begin{equation}\label{eq:98}
  \widetilde\psi({\c},\dot{\c})=\widehat\psi({\c}),
\end{equation}
and \eqref{eq:112} becomes: 
\begin{equation}
  \label{eq:17}
(\partial_{\c}\widehat\psi({\c})-\widetilde\mu({\c},\dot{\c}))\dot{\c}+\widetilde{\mathbf h}({\c},\dot{\c},\mathbf g)\cdot\mathbf g\le 0.
\end{equation}
Arguing as in \cite[Appendix]{BertsPV2001AMPA}, we set $\mathbf g=\mathbf 0$ in the above inequality, and we reckon that the function $\lambda\mapsto \partial_{\c}\widehat\psi({\c})-\widetilde\mu({\c},\lambda)$ changes its sign at $\lambda=0$. As this function is smooth, we conclude that it vanishes for $\lambda=0$, and hence:
\begin{equation}\label{eq:97}
  \widetilde\mu({\c},0)=\partial_{\c}\widehat\psi({\c}).
\end{equation}
It follows from \eqref{eq:97} that exists $\widehat\beta({\c},\dot{\c})$ such that
\begin{equation}
  \label{eq:28}
  \widetilde\mu({\c},\dot{\c})=\partial_{\c}\widehat\psi({\c})+\widehat\beta({\c},\dot{\c})\dot{\c}.
\end{equation}
Next, we define
\begin{equation}
  \label{eq:8}
  \widehat{\mathbf h}({\c},\mathbf g)=\widetilde{\mathbf h}({\c},0,\mathbf g),
\end{equation}
and 
\begin{equation}
  \label{eq:18}
  \widehat{\mathbf m}({\c},\dot{\c},\mathbf g)=
\begin{cases}
 \partial_{\dot{\c}}\widetilde{\mathbf h}({\c},0,\mathbf g)&\quad\textrm{if }{\dot\c}=0,\\
 \displaystyle\frac{ \widetilde{\mathbf h}({\c},\dot{\c},\mathbf g)-\widehat{\mathbf h}({\c},\mathbf g)}{\dot{\c}}&\textrm{otherwise}.
\end{cases}
\end{equation}
Then, we can write
\begin{equation}
  \label{eq:29}
  \widetilde{\mathbf h}({\c},\dot{\c},\mathbf g)=\widehat{\mathbf h}({\c},\mathbf g)+\widehat{\mathbf m}({\c},\dot{\c},\mathbf g)\dot{\c}.
\end{equation}
On taking \eqref{eq:28} and \eqref{eq:29} into account, we arrive at the following form of the dissipation inequality:
\begin{equation}
  \label{eq:117}
-\widehat\beta({\c},\dot{\c})\dot{\c}+\mathbf g\cdot\widehat{\mathbf m}({\c},\dot{\c},\mathbf g)\dot{\c}+\widehat{\mathbf h}({\c},\mathbf g)\cdot\mathbf g\le 0.
\end{equation}
On setting $\dot{\c}=0$ and on insisting that \eqref{eq:117} holds for whatever choice of $\mathbf g$, we recover the representation \eqref{eq:24} for $\widehat{\mathbf h}$. 

Now, on taking into account \eqref{eq:98}, \eqref{eq:28}, \eqref{eq:24}, and \eqref{eq:29}, we can rewrite the constitutive equations \eqref{eq:111} as:
\begin{equation}
  \label{eq:30}
  \psi=\widehat\psi({\c}),\quad \mu=\partial_{\c}\widehat\psi({\c})+\widehat\beta({\c},\dot{\c})\dot{\c},\quad\mathbf h=-\widehat{\mathbf M}({\c},\mathbf g)\mathbf g+\dot{\c}\,\widehat{\mathbf m}({\c},\dot{\c},\mathbf g).
\end{equation}
We are now in position to interpret the parabolic regularized equation \eqref{eq:13} as a consequence of the following specifications in \eqref{eq:30}:
\begin{equation}
  \label{eq:19}
  \widehat\beta({\c},\dot{\c})=\beta,\qquad \widehat{\mathbf M}(u,\mathbf g)=\alpha\mathbf I,\qquad\widehat{\mathbf m}({\c},\dot{\c},\mathbf g)=\mathbf 0,
\end{equation}
where the positivity of the constants $\alpha$ and $\beta$ guarantees that the dissipation inequality is never violated.

\section{Existence of weak solutions in the viscous case}\label{sec:appendix}
In this section we consider the following system:
\begin{subequations}
\begin{equation}\label{system}
\begin{split}
  &\dot{\c}+\textrm{div}(\alpha\nabla\mu)=0, \\ 
  & \mu=\beta\dot{\c}+\widehat\psi'({\c}),
\end{split}
\end{equation}
with $\alpha$ and $\beta$ strictly positive constants. We obtain system \eqref{system} by combining the mass-balance equation \eqref{eq:1} with the constitutive prescription \eqref{eq:2}  and the viscous regularization \eqref{eq:53}. For $\Omega\subset\mathbb R^3$ a domain with smooth boundary $\Gamma$, and for $T>0$,  we prove existence of weak solutions to \eqref{system} in the parabolic domain $Q=\Omega\times(0,T)$ with the initial condition
\begin{equation}
u(0)=u_0\qquad \textrm{ in }\Omega,
\end{equation}
and with the (possibly) non-homogeneous Neumann condition:
\begin{align}\label{eq:93}
  \nabla\mu\cdot \mathbf n=h\quad\textrm{ on }\Sigma,
\end{align}
\end{subequations}
imposed on the parabolic boundary $\Sigma=\Gamma\times(0,T)$, with $\mathbf n$ the outward unit normal on $\Gamma$. The treatment of the non-homogeneous Dirichlet condition for $\mu$, can be found in \cite{BonetCT2015?}.

The proof proposed by Novick-Cohen and Pego in \cite{NovicP1991TAMS} is based on regarding \eqref{system} as an ODE in a suitable Banach space. Here we use a different approach, based on the Galerkin method and on compactness arguments (see for instance \cite{Lions1969} or \cite{Roubi2013Nonlinear}). Our main purpose here is to illustrate where the viscous regularization comes in handy as far as existence of solutions is concerned. 

For typographical convenience, we henceforth write $\psi$ in place of $\widehat\psi$; furthermore, we use the symbol $C_i$ to denote a generic positive constant that depends on the index $i$; moreover, given $p>1$ we denote by
\begin{equation}\label{eq:47}
  p'=\frac{p}{p-1}
\end{equation}
the conjugate H\"older exponent of $p$; moreover, as a rule, we do not relabel subsequences. Following standard notation, for $B$ a Banach space and $r\ge 1$ we denote by $L^r(0,T;B)$ the $L^r$--Bochner space of $B$-valued functions defined on the interval $(0,T)$, and by $H^1(0,T;B)$ the corresponding Sobolev-Bochner space. 
For the definition of these spaces, we make reference to Sections 1.5 and 7.1 of \cite{Roubi2013Nonlinear}. When elements of these functions spaces shall be involved in the definition of double integrals both with respect to the time variable $t$ and space variable $x$, the dependence on the latter shall be left tacit.
\medskip

\noindent\textbf{Assumptions.} We assume that the free energy is twice continuously differentiable, and that its second derivative be bounded from below, namely: 
\begin{equation}
  \label{eq:39}
  \color{blue}\psi''(r)\ge -M_0
\end{equation}
for some $M_0>0$, for all $r\in\mathbb R$.
Moreover, \color{blue}we assume that \color{black}there exist  $M_i>0$, $i=1\dots 5$ and $p\in[2,6)$ such that, for all $r\in\mathbb R$,
  \begin{equation}
    \label{eq:27}
    -M_1+M_2|r|^{p}\le \psi(r)\le M_3+M_4|r|^{p}
  \end{equation}
and
\begin{equation}
  \label{eq:44}
  |\psi'(r)|\le M_5(1+|r|^{p-1}).
\end{equation}
Finally, we assume:
\begin{subequations}\label{eq:75}
\begin{align}
 &h\in L^{p}(0,T;L^2(\Gamma)),\label{eq:75a}
\end{align}
and\color{black}
\begin{align}\label{eq:83}
 &{\c}_0\in H^1(\Omega),
\end{align}
\end{subequations}
with $p$ the same exponent as in \eqref{eq:27} and \eqref{eq:44}. Under the aforementioned assumptions, we are going prove the following:
\begin{theorem}[Existence of weak solutions to \eqref{system}]\label{thm:1}
There exist 
\begin{align*}
  \label{eq:46}
&{\c}\in L^\infty(0,T;H^1(\Omega))\cap H^1(0,T;L^2(\Omega)),\\
&\mu\in L^2(0,T;H^1(\Omega)),
\end{align*}
such that
\begin{subequations}\label{eq:68}
\begin{align}
&\int_\Omega \color{blue}\big(\dot{\c}(t) v+\alpha\nabla\mu(t)\cdot\nabla v\big)\color{black}\,{\rm d}x=\int_{\Gamma} h(t)v\,{\rm d}\Gamma\quad\forall v\in H^1(\Omega), \text{ for a.a. } t\in(0,T),\\
&\mu=\color{blue}\beta\color{black}\dot{\c}+\psi'({\c})\quad \textrm{a.e. in } Q,\label{eq:52}\\
&{\c}(0)={\c}_0 \quad \textrm{a.e. in }\Omega.\label{eq:52bis}
\end{align}
\end{subequations}
\end{theorem}

\begin{proof} For the sake of readability, we split the proof into a sequence of intermediate steps.
\medskip

\emph{Step 1. Selection of a basis.} We denote by $\{v_n\}_{n=1}^\infty$ the eigenfunctions of the Laplace operator with Neumann boundary conditions:\marginpar{\tiny ELLIPTIC REGULARITY}
\begin{subequations}\label{eq:1312}
\begin{align}
  -\Delta v_n=\lambda_n v_n\quad\textrm{in }\Omega,\\
  \nabla v_n\cdot\mathbf n=0\quad\textrm{on }\Gamma,
\end{align}
\end{subequations}
and we define
\begin{equation}
V_n={\rm span}(v_1,\dots v_n).
\end{equation}
The set $\cup_n V_n$ is dense in $H^1(\Omega)$ and in $L^p(\Omega)$, where $p$ is the exponent in \eqref{eq:27} and \eqref{eq:44}. \color{blue}Without any loss of generality, we can assume that the collection $\{v_n\}_{n=1}^\infty$ is an orthonormal system \color{black}for $L^2(\Omega)$, that is:
\begin{equation}
  \label{eq:42}
  \int_\Omega v_nv_k\,{\rm d}x=\delta_{nk},
\end{equation}
and that the first element of the basis is the constant function: $v_1(x)=1/{\sqrt{|\Omega|}}$. 
We approximate the \emph{initial datum} ${\c}_0$ through a sequence $\{{\c}_{0,n}\}_{n=1}^\infty$ of functions having the form
 \begin{equation}
  \label{eq:54}
  {\c}_{0,n}(x)=\sum_{i=1}^n a_{0,n,i}v_i(x)
\end{equation}
such that 
\begin{equation}
  \label{eq:76}
{\c}_{0,n}\to {\c}_0 \quad\textrm{ strongly in }H^1(\Omega),
\end{equation}
with convergence holding also in $L^p(\Omega)$ because of the Sobolev embedding. Note that, thanks to our assumption \eqref{eq:83} on the initial datum  and to our growth assumption \eqref{eq:27} on the coarse-grain free energy $\psi$, we have
\begin{equation}
  \label{eq:55}
  \psi({\c}_{0,n})\to\psi({\c}_0)\quad\textrm{ strongly in }L^1(\Omega).
\end{equation}
Thus, in particular, there exists a positive constant $C_0$ such that:
\begin{equation}\label{eq:84}
  \int_\Omega \psi({\c}_{0,n}){\rm d}x\le C_0.
\end{equation}
\medskip

\emph{Step 2. Construction of an approximating sequence.} We introduce the \emph{Galerkin approximations}
\begin{subequations}\label{eq:3}
\begin{align}
  &{\c}_n(x,t)=\sum_{i=1}^n a_{n,i}(t) v_i(x),\label{eq:3a}\\
 &\mu_n(x,t)=\sum_{i=1}^n b_{n,i}(t)v_i(x),\label{eq:3b}
\end{align}
\end{subequations}
and we look for time-dependent coefficients $(a_{n,i}(t),b_{n,i}(t))_{i=1\dots n}$ such that
\begin{subequations}\label{eq:4}
\begin{align}
 & \int_\Omega \color{blue}\big(\color{black}\dot{\c}_n v+\alpha\nabla\mu_n\cdot\nabla v\color{blue}\big)\color{black}\,{\rm d}x=\int_{\Gamma} h v\,{\rm d}\Gamma \qquad\textrm{for all } v\in V_n,\label{eq:4a}\\
 & \int_\Omega \mu_nv \,{\rm d}x=\int_\Omega \big(\beta\dot{\c}_n+\psi'({\c}_n)\big)v\,{\rm d}x=0\qquad\textrm{for all }v\in V_n,\label{eq:4b}
\end{align}
\end{subequations}
at all times, and 
\begin{equation}
  \label{eq:49}
   {\c}_n(0)={\c}_{0,n}.
\end{equation}
To this aim, by testing \eqref{eq:4a} and \eqref{eq:4b} by $v_i$, we find that the system \eqref{eq:4} is equivalent to:
\begin{equation}\label{eq:69}
\left.
\begin{array}{l}
 \dot a_{n,i}(t)+\alpha\lambda_i b_{n,i}(t)=H_i(t),\\[1em]
 b_{n,i}(t)=\beta\dot a_{n,i}(t)+G_{n,i}(a_{n,1}(t),\dots a_{n,n}(t)),
\end{array}
\right\}i=1,\dots,n,
\end{equation}
with initial conditions:
\begin{equation}\label{eq:74}
   a_{n,i}(0)=a_{0,n,i},\qquad i=1,\dots,n,
\end{equation}
where $H_{i}(t)=\int_{\Gamma} h(x,t) v_i(x){\rm d}\Gamma$ and where the functions $G_{n,i}:\mathbb R^n\to\mathbb R$ are defined by $G_{n,i}(r_1,\dots, r_n)=\int_\Omega \psi'\Big(\sum_{j=1}^n r_j v_j(x)\Big) v_i(x)\,{\rm d}x$. Since the eigenvalues $\lambda_i$ are non negative, the two groups of equations in \eqref{eq:69} can be combined to obtain a $n\times n$ system of ordinary differential equations:
\begin{align}\label{eq:23}
  \dot a_{n,i}(t)=\frac{H_i(t)-\alpha\lambda_iG_{n,i}(a_{n,1}(t),\dots a_{n,n}(t))}{1+\alpha\beta\lambda_i},\quad i=1\dots n,
\end{align}
which together with the initial conditions \eqref{eq:74} defines a Cauchy problem. 

By the smoothness of $\psi$ and of the basis functions, the functions $G_{n,i}$ are smooth as well, and for each $i\in 1\dots n$ the right-hand side in \eqref{eq:23} is locally Lipschitz continuous. Thus, according to standard ODE theory, the Cauchy problem \eqref{eq:74}--\eqref{eq:23} has a unique solution on a non-empty interval $(0,T_n)$. If $T_n<T$, we must show that the solution admits a continuation up to time $T$. We achieve this goal in the next step by deriving bounds on the coefficients in \eqref{eq:3} based on an estimate  which mimicks the natural energetic estimate for the original system \eqref{system}.\medskip

\emph{Step 3. Energetic estimate.} 
We consider a subinterval $(0,t)$ of $(0,T_n)$, and for each $s\in (0,t)$ we test \eqref{eq:4a} and \eqref{eq:4b} with $\mu_n(s)$ and $-\dot{\c}_n(s)$, respectively. On adding the resulting equations and on integrating over $(0,t)$ we obtain:
\begin{multline}
  \label{eq:10}
  \int_\Omega \psi({\c}_n(t))\,{\rm d}x+\int_0^t\!\!\int_\Omega \color{blue}\big(\color{black}\alpha|\nabla\mu_n(s)|^2+\beta\dot{\c}_n^2(s)\color{blue}\big)\color{black}\,{\rm d}x{\rm d}s\\
= \int_0^t\!\!\int_{{\Gamma}}h(s)\mu_n(s)\,{\rm d}\Gamma{\rm d}s+
\int_\Omega \psi({\c}_{0,n}){\rm d}x.
\end{multline}
Thanks to the coercivity assumption in \eqref{eq:27}, and thanks to \eqref{eq:84}, we have
\newcounter{constant}
\setcounter{constant} 0
\begin{multline}\label{eq:57}
  M_2\int_\Omega |u(t)|^p\,{\rm d}x+\int_0^t\!\!\int_\Omega \big(\alpha|\nabla\mu_n(s)|^2+\beta\dot{\c}_n^2(s)\big){\rm d}x{\rm d}s
\\ 
\le \underbrace{M_1|\Omega|+C_0}_{\displaystyle C_{\addtocounter{constant}{1}\theconstant}}+\int_0^t\!\!\int_{{\Gamma}}h(s)\mu_n(s){\rm d}\Gamma{\rm d}s,
\end{multline}
where $M_1$ and $C_0$ are the constants in, respectively, \eqref{eq:27} and \eqref{eq:84}. Using, in the order, the trace theorem, Poincare's inequality, and Young's inequality, we have, at each time (whose specification we omit to keep our notation terse),
\begin{align}
  \int_{{\Gamma}}h\mu_n\,{\rm d}\Gamma&\le \|h\|_{L^2(\Gamma)}\|\mu_n\|_{L^2(\Gamma)}
\nonumber
\\
&\le C_{\addtocounter{constant}{1}\theconstant}\|h\|_{L^2({\Gamma})}\|\mu_n\|_{H^1(\Omega)}\nonumber
\\
&\le C_{\addtocounter{constant}{1}\theconstant} \|h\|_{L^2({\Gamma})}\left(\|\nabla\mu_n\|_{L^2(\Omega)}+\Big|\int_\Omega \mu_n\,{\rm d}x\Big|\right)\nonumber\\
&\le C_{\addtocounter{constant}{1}\theconstant}\left(\frac 1{\delta^p}\|h\|^p_{L^2({\Gamma})}+{\delta^{p'}}\left(\|\nabla\mu_n\|^{p'}_{L^2(\Omega)}+\Big|\int_\Omega \mu_n\,{\rm d}x\Big|^{p'}\right)\right),
\label{eq:38}
\end{align}
where we observe that (recall \eqref{eq:47} and the assumption on $p$ in the sentence preceding \eqref{eq:27}):
\begin{equation}\label{eq:40}
  p'\le 2.
\end{equation}
On choosing $v=1$ as test in \eqref{eq:4b}, we obtain $\int_\Omega \mu_n{\rm d}x=\int_\Omega\big(\alpha\dot{\c}_n+\psi'({\c}_n)\big){\rm d}x$. Then, on taking the growth assumption \eqref{eq:44} into account, we get $\Big|\int_\Omega \mu_n\,{\rm d}x\Big|\le  C_{\addtocounter{constant}{1}\theconstant}\Big(1+\int_\Omega \big(|\dot{\c}_n|+|{\c}_n|^{p-1}\big)\,{\rm d}x\Big)$, and hence:
\begin{align}
  \label{eq:37}
  \Big|\int_\Omega \mu_n\,{\rm d}x\Big|^{p'}\le {} & C_{\addtocounter{constant}{1}\theconstant}\Big(1+\int_\Omega \big(|\dot{\c}_n|^{p'}+|{\c}_n|^{p}\big)\,{\rm d}x\Big).
\end{align} 
Thus, \eqref{eq:38} and \eqref{eq:37}, along with our assumption \eqref{eq:75a} on $h$, and \eqref{eq:40}, yield, for every $\delta>0$,
\begin{equation}\label{eq:86}
\begin{aligned}
  \int_{{\Gamma}}h\mu_n\,{\rm d}\Gamma&\le \frac {C_4}{\delta^p}\|h\|_{L^2({\Gamma})}^p+\delta^{p'}C_{\addtocounter{constant}{1}\theconstant}\Big(1+\|\nabla\mu_n\|^{p'}_{L^2(\Omega)}+\int_\Omega \big(|\dot{\c}_n|^{p'}+|{\c}_n|^{p}\big){\rm d}x\Big)\\
 &\le \frac{C_{\addtocounter{constant}{1}\theconstant}}{\delta^p}+\delta^{p'}C_{\addtocounter{constant}{1}\theconstant}\Big(1+\int_\Omega \big(|\nabla\mu_n|^{2}+|\dot{\c}_n|^{2}+|{\c}_n|^{p}\big){\rm d}x\Big).
\end{aligned}
\end{equation}
By gathering \eqref{eq:57} and \eqref{eq:86}, we obtain
\begin{multline}\label{eq:72}
  M_2\int_\Omega |u_n(t)|^p\,{\rm d}x+\int_0^t\!\!\int_\Omega \big(\alpha|\nabla\mu_n(s)|^2+\beta\dot{\c}_n^2(s)\big)\,{\rm d}x{\rm d}s
\\
\le C_1+\frac {C_8}{\delta^{p}}+\delta^{p'}C_{9}\int_0^t\!\!\int_\Omega \big(|\nabla\mu_n(s)|^2+|\dot{\c}_n(s)|^{2}+|{\c}_n(s)|^{p}\big){\rm d}x{\rm d}s,
\end{multline}
whence, on taking $\delta$ sufficiently small,
\begin{equation}
  \label{eq:10-bis}
 M_2\int_\Omega |{\c}_n(t)|^p{\rm d}x+\frac 12 \int_0^t\!\!\int_\Omega \big(\alpha|\nabla\mu_n(s)|^2+\beta\dot{\c}_n^2(s)\big){\rm d}x{\rm d}s\le C_{\addtocounter{constant}{1}\theconstant}\Big(1+\int_0^t\!\!\int_\Omega|{\c}_n(s)|^{p}{\rm d}x{\rm d}s\Big).
\end{equation}
By using Gronwall's lemma in \eqref{eq:10-bis}, we obtain that the coefficients $a\nn_i(t)$ are bounded for $t\in (0,T_n)$ by a constant that does not depend on $n$. It then follows that the solution of \eqref{eq:23} can be continued up to time $T$ and that the bounds:
\begin{subequations}\label{eq:2211}
\begin{align}
&\|{\c}_n\|_{L^\infty(0,T;L^p(\Omega))}\le C_{\addtocounter{constant}{1}\theconstant},\label{aa}\\
&\|\dot{\c}_n\|_{L^2(Q)}\le C_{\addtocounter{constant}{1}\theconstant},\label{bb}
\end{align}
\end{subequations}
hold uniformly with respect to $n$. Moreover, we have that $\|\nabla\mu_n\|_{L^2(Q)}$ is uniformly bounded, a fact that, together with \eqref{eq:37} and \eqref{eq:2211} yields
\begin{equation}
  \label{eq:43}
  \|\mu_n\|_{L^2(0,T;H^1(\Omega))}\le C_{\addtocounter{constant}{1}\theconstant}.
\end{equation}
\medskip

\emph{Step 4. Limit passage in the diffusion equation.}
The estimates \eqref{bb} and \eqref{eq:43} are what we need to pass to the limit in the linear equation \eqref{eq:4a} to obtain \eqref{eq:46}. Indeed, thanks to these estimates there exists a subsequence (not relabeled) such that
\begin{align}
  \label{eq:44a}
& {\c}_n\to{\c}\quad\textrm{weakly in } H^1(0,T;L^2(\Omega)),\\
 & \mu_n\to \mu \quad\textrm{weakly in }L^2(0,T;H^1(\Omega)).\label{eq:44b}
\end{align}
Now, given any $k\le n$, we have that $v_k$ is a legal test in \eqref{eq:4a}. Thus, for every for every $t\in(0,T)$ we can consider:
\begin{equation}
  \label{eq:32}
  0=\frac 1 h\int_{t}^{t+h}\int_\Omega \big(\dot {\c}_n(s) v_k+\alpha\nabla\mu_n(s)\cdot\nabla v_k\big){\rm d}x{\rm d}s=0\qquad\forall n\ge k,
\end{equation}
for all $0<h<T-t$. On letting first $n\to\infty$ and then $h\to 0$ in \eqref{eq:32}, the Steklov average in \eqref{eq:32} converges for a.e. $t\in(0,T)$:
\begin{equation}
  \label{eq:32b}
  \int_\Omega \dot {\c}(t) v_k+\alpha\nabla\mu(t)\cdot\nabla v_k=0\qquad\forall t\in (0,T)\setminus E_k,
\end{equation}
where $\textrm{meas}(E_k)=0$ is a set that depends on the particular $k$. Now, given that \eqref{eq:32b} holds for $k$ arbitraty, we conclude that 
\begin{equation}
  \label{eq:32b}
  \int_\Omega \big(\dot {\c}(t) v+\alpha\nabla\mu(t)\cdot\nabla v\big){\rm d}x=0\quad\forall v\in \cup_k V_k,\quad \forall t\in (0,T)\setminus \cup_k E_k.
\end{equation}
Since $\cup_k V_k$ is dense in $H^1(\Omega)$, and since the set $\cup_k E_k$ has zero measure, we conclude that
 \eqref{eq:46} holds true.\medskip

\emph{Step 5. Gradient estimate on concentration.}
In order to derive \eqref{eq:52} by passing to the limit in \eqref{eq:4b} weak convergence does not suffice because of the possibly nonlinear term $\psi'({\c}_n)$. We therefore seek extra compactness properties for ${\c}_n$, a piece of information that we shall gather from an additional estimate. In order to get this estimate, we first choose ${\c}_n$ as tests in \eqref{eq:4a} to get, for every $s\in[0,T]$,
\begin{subequations}
\begin{align}\label{eq:87}
&\int_\Omega \color{blue}\big(\dot{\c}_n(s){\c}_n(s)+\alpha\nabla\mu_n(s)\cdot\nabla {\c}_n(s)\big)\color{black}\,{\rm d}x=\int_{\Gamma} h(s){\c}_n(s)\,{\rm d}x;
\end{align}
then, we take $v= -\alpha \Delta {\c}_n$ in \eqref{eq:4b} and we resort to by-parts integration to arrive at:
\begin{align}\label{eq:88}
&\alpha\int_\Omega \nabla\mu_n(s)\cdot\nabla{\c}_n(s){\rm d}x=\color{blue}\alpha\int_\Omega \big(\beta\color{black}\nabla\dot{\c}_n(s)\cdot\nabla{\c}_n(s)+\psi''({\c}_n(s))|\nabla{\c}_n(s)|^2\big){\rm d}x,
\end{align}
\end{subequations}
which again holds for all $s\in[0,T]$; next, we subtract \eqref{eq:88} from \eqref{eq:87} and we integrate with respect to $s$ on $(0,t)$ to get:
\begin{multline}
  \label{eq:21}
  \frac 1 2\int_\Omega \big({\c}_n^2(t)+\alpha\beta|\nabla{\c}_n(t)|^2\big){\rm d}x+\alpha\int_0^t\!\!\int_\Omega \psi''({\c}_n(s))|\nabla{\c}_n(s)|^2\,{\rm d}x{\rm d}s\\
=\frac 1 2\int_\Omega\big( {\c}_n^2(x,0)+\alpha\beta|\nabla{\c}_n(x,0)|^2\big){\rm d}x+
\int_0^t\!\!\int_\Gamma h(s){\c}_n(s)\,{\rm d}\Gamma{\rm d}s
\end{multline}
It is precisely at this point that the viscous regularization comes in handy. Indeed, with the aid of Assumption \eqref{eq:39} we obtain from \eqref{eq:21} the following inequality:
\begin{multline}\label{eq:85}
  \frac 1 2\int_\Omega \big({\c}_n^2(t)+\alpha\beta|\nabla{\c}_n(t)|^2\big){\rm d}x\le
\alpha M_0\int_0^t\!\!\int_\Omega |\nabla{\c}_n(s)|^2\,{\rm d}s+\int_0^t\!\!\int_\Gamma h(s){\c}_n(s)\,{\rm d}\Gamma{\rm d}s\\
+\frac 1 2\int_\Omega\big( {\c}_n^2(x,0)+\alpha\beta|\nabla{\c}_n(x,0)|^2\big){\rm d}x.
\end{multline}
Next, by making use of \eqref{eq:75a} and by invoking the trace theorem, we estimate:
\begin{align}
  \int_0^t\!\!\int_\Gamma h(s){\c}_n(s)\,{\rm d}\Gamma{\rm d}s&\le  \frac 12 \int_0^t\!\!\int_\Gamma h^2(s)\,{\rm d}\Gamma{\rm d}s+\frac 12 \int_0^t\!\!\int_\Gamma {\c}^2_n(s)\,{\rm d}\Gamma{\rm d}s\nonumber
\\
&\le C_{\addtocounter{constant}{1}\theconstant}\left(1+\int_0^t\!\!\|{\c}^2_n(s)\|^2_{H^1(\Omega)}\,{\rm d}s\right).\label{eq:20}
\end{align}
On combining the inequalities \eqref{eq:85} and \eqref{eq:20}, and by recalling \eqref{eq:76}, we arrive at:
\begin{equation}\label{eq:85}
  \|{\c_n}(t)\|_{H^1(\Omega)}^2\le C_{\addtocounter{constant}{1}\theconstant}\left(1+\int_0^t\|{\c_n}(s)\|_{H^1(\Omega)}^2\,{\rm d}s\right).
\end{equation}
Now, the application of Gronwall's lemma to \eqref{eq:85} yields:
\begin{align}
 &\|{\c}_n\|_{L^\infty(0,T;H^1(\Omega))}\le C_{\addtocounter{constant}{1}\theconstant},\label{eq:233}
\end{align}
a bound that we are going to use in the next step to pass to the limit in the nonlinear equation \eqref{eq:4b}. \medskip

\emph{Step 6. Limit passage in the equation for chemical potential.} Since $p<6$ and $\Omega$ is a three-dimensional domain, $H^1(\Omega)$  is compactly embedded in $L^p(\Omega)$. Consequently, we can use the bounds \eqref{eq:233} and \eqref{bb}, along with the Aubin-Lions lemma in the form of Corollary 4 of \cite{Simon1987AMPA4}, to establish the following convergence result:
\begin{equation}
  \label{eq:45}
  {\c}_n\to{\c}\quad\textrm{in }C([0,T];L^p(\Omega)).
\end{equation}
We note on passing that \eqref{eq:45}, along with with \eqref{eq:76}, entails the initial condition \eqref{eq:52bis}.

Now, on rewriting Assumption \eqref{eq:44} as $|\psi'(r)|\le M_5(1+|r|^{p/p'})$, we immediately see that standard results concerning mappings between Lebesgue spaces (see for instance Theorem 1.27 in \cite{Roubi2013Nonlinear}) do apply, and we can conclude that the Nemytski\v{\i}  mapping $u\mapsto \psi'(u)$ is bounded and strongly continuous from $L^p(\Omega)$ to $L^{p'}(\Omega)$. Thus, the convergence
\begin{equation}
  \label{eq:50}
  \psi'({\c}_n)\to\psi'({\c})\quad\textrm{in }C([0,T];L^{p'}(\Omega))
\end{equation}
follows from \eqref{eq:45}. Now, pick arbitrary functions $a\in C([0,T])$ and $\phi\in L^p(\Omega)$, and let $\phi_n\in V_n$ be a sequence of linear combinations of basis functions such that 
\begin{equation}\label{eq:90}
  \phi_n\to\phi\quad\textrm{strongly in }L^p(\Omega).
\end{equation}
On testing \eqref{eq:4b} by $a(t)\phi_n$ at each particular time level $t\in(0,T)$ and on integrating over $(0,T)$ we get:
\begin{align}\label{eq:89}
 \int_0^T\!\!\!\int_\Omega \big(\mu_n(t)-\beta\dot{\c}_n(t)+\psi'({\c}_n(t))\big)a(t)\phi_n\,{\rm d}x{\rm d}t=0.
\end{align}
By relying on the weak convergences \eqref{eq:44a} and \eqref{eq:44b}, and on the strong convergences \eqref{eq:50} and \eqref{eq:90}, we can pass to the limit in \eqref{eq:89} to get:
\begin{align}\label{eq:91}
 \int_0^T \!\!\Big(\int_\Omega \big(\mu(t)-\beta\dot{\c}(t)-\psi'({\c}(t))\big)\phi\,{\rm d}x\Big)a(t){\rm d}t=0
\end{align} 
By the arbitrariness of the function $t\mapsto a(t)$ we conclude that
\begin{equation}\label{eq:99}
\int_\Omega \big(\mu(t)-\beta\dot{\c}(t)-\psi'({\c}(t))\big)\phi\,{\rm d}x\qquad \textrm{for a.e. }t\in(0,T)
. 
\end{equation}
Eventually, the arbitrariness of $\phi$ in \eqref{eq:99} entails \eqref{eq:52}.
\medskip

\emph{Step 7. Verification of the initial condition.} In order to show that the initial condition \eqref{eq:52bis} is satisfied, we begin by observing that, at the approximation level $n$, 
\begin{align}\label{eq:70}
  \int_0^T\!\!\!\int_\Omega\dot{\c}_n\phi\,{\rm d}x{\rm d}t=-\int_0^T\!\!\!\int_\Omega {\c}_n\dot\phi\,{\rm d}x{\rm d}t-\int_\Omega {\c}_{0,n}\phi(0)\hskip 0.2pt{\rm d}x
\end{align}
for every smooth test function $\phi\in C^\infty(\overline{Q})$ such that $\phi(T)=0$. By passing to the limit in \eqref{eq:70}, using the weak convergence of the sequences in \eqref{eq:44b} and the convergence of the initial conditions for the stated in \eqref{eq:76}, we arrive at:
\begin{align}\label{eq:51}
 \int_0^T\!\!\!\int_\Omega \dot{\c}\,\phi\,{\rm d}x{\rm d}t=-\int_0^T\!\!\!\int_\Omega {\c}\,\dot\phi\,{\rm d}x{\rm d}t-\int_\Omega {\c}_{0}\hskip 0.5pt\phi(0)\hskip 0.2pt{\rm d}x.
\end{align}
On the other hand, since $\dot u\in L^2(Q)$, we obtain, through integration by parts,
\begin{align}
  \hskip 1em\int_0^T\!\!\!\int_\Omega\dot{\c}\,\phi\,{\rm d}x{\rm d}t=-\int_0^T\!\!\!\int_\Omega {\c}\,\dot\phi\,{\rm d}x{\rm d}t-\int_\Omega {\c}(0)\hskip 0.2pt\phi(0)\hskip 0.5pt{\rm d}x,
\end{align}
whence, by a comparison with \eqref{eq:51}, 
\begin{align}\label{eq:71}
  \int_\Omega {\c}(0)\phi(0)\hskip 0.5pt{\rm d}x=\int_\Omega{\c}_{0}\, \phi(0)\hskip 0.5pt{\rm d}x,
\end{align}
and thence the initial condition \eqref{eq:52bis}, by the arbitrariness of $\phi(0)$.
\end{proof}

\section{Additional remarks}
\noindent\textbf{Non--smooth dissipation potentials and hysteresis.} Our discussion in Section 3 suggests that the class of viscous regularizations described by \eqref{eq:53} may be replaced with more encompassing prescriptions for $\mu$, provided that the result is consistent with the dissipation inequality. For example, one may want to drop the smoothness assumption we have made so far and consider the \emph{inclusion:}
\begin{equation}
  \label{eq:35}
   \mu-\widehat\psi'({\c})\in\partial\widehat\zeta(\dot{\c})
 \end{equation}
with $\widehat\zeta$ a proper, lower-semicontinuous, convex \emph{dissipation potential}. 
We here restrict attention to the special case:
\begin{equation}
   \label{eq:36}
   \widehat\zeta(r)=\frac 1 2\beta r^2+\gamma|r|,
 \end{equation}
with $\beta>0$ a \emph{viscous-damping parameter} and $\gamma>0$ a \emph{threshold parameter}. In this case, the subdifferential set of $\widehat\zeta$ at $r$ is given by:
 \begin{align}
   \partial\widehat\zeta(r)=
\begin{cases}
\displaystyle\Big\{\beta r+\gamma \frac{r}{|r|}\Big\}\ 
\qquad\textrm{if }r\neq 0,\\[0.5em]
[-\gamma,+\gamma]\, \qquad\qquad\textrm{if }r=0.
\end{cases}
 \end{align}
Accodingly, the inclusion \eqref{eq:35} is equivalent to:
\begin{subequations}\label{eq:33bb}
 \begin{align}
   &\mu=\widehat\psi'({\c})+\beta\dot{\c}+\gamma \frac{\dot{\c}}{|\dot{\c}|}\quad\textrm{if }\dot{\c}\neq 0,\\   
   &\mu-\widehat\psi'({\c})\in [-\gamma,\gamma]\ \, \qquad\textrm{if }\dot{\c}=0.
 \end{align}
\end{subequations}
The prescriptions \eqref{eq:33bb}, when combined with the PDE that results from \eqref{eq:1} and \eqref{eq:2}, namely,
\begin{equation}
  \label{eq:48}
  \dot{\c}=\alpha\Delta\mu,
\end{equation}
together with appropriate boundary conditions, yield an  evolution problem in the unknowns ${\c}$ and $\mu$. The analysis of this problem, which is more involved than the case of viscous regularization, shall be presented in \cite{BonetCT2015?}. 

We expect that the solutions of \eqref{eq:33bb}--\eqref{eq:48}  display \emph{hysteretic behavior}, even if the free energy mapping $\widehat\psi$ is convex. As an example to support this point, we consider that the free-energy mapping is quadratic:
\begin{equation}\label{eq:92}
  \widehat\psi(r)=\frac 12 Kr^2,
\end{equation}
so that \eqref{eq:33bb} becomes:
\begin{subequations}\label{eq:100}
 \begin{align}
   &\mu=K{\c}+\beta\dot{\c}+\gamma \frac{\dot{\c}}{|\dot{\c}|}\quad\textrm{if }\dot{\c}\neq 0,\\   
   &\mu-K{\c}\in [-\gamma,\gamma]\ \, \qquad\textrm{if }\dot{\c}=0.
 \end{align}
\end{subequations}
We impose a homogeneous initial condition:
\begin{equation}
  u(x,0)=0\qquad \textrm{for all }x\in\Omega.
\end{equation}
Moreover, instead of prescribing the flux at the boundary as in Section 4, we impose the Dirichlet condition:
\begin{equation}\label{eq:81}
  \mu(x,t)=Az\Big(\frac t\tau\Big)\qquad \textrm{for all } x\in\Gamma,
\end{equation}
where $A>0$ is the \emph{half-amplitude}, $\tau>0$ is the \emph{period} and $z(s)$ the \emph{zigzag function} (\emph{cf.} Fig 1.(a) below) defined on $[0,1)$ by:
\begin{equation}\label{eq:113}
 z(s)=
\left\{
\begin{array}{ll}
 4s&\textrm{if }s\in [0,1/4),\\[0.5em]
 2-4s&\textrm{if }s\in [1/4,3/4),\\[0.5em]
 -4+4s&\textrm{if }s\in [3/4,1),
\end{array}
\right.
\end{equation}
and extended periodically to $\mathbb R$. 

The particular form of the boundary datum prompts us to rewrite \eqref{eq:33bb} and \eqref{eq:48} using the \emph{dimensionless time}: 
\begin{equation}
  s=\frac t \tau
\end{equation}
in place of $t$ as independent variable. Accordingly, all fields are henceforth functions of $x$ and $s$, and a superimposed dot denotes the derivative with respect to $s$. Then, \eqref{eq:33bb} and \eqref{eq:48} become, respectively,
\begin{subequations}\label{eq:77}
 \begin{align}
   &\mu=K{\c}+\frac\beta\tau\dot{\c}+\gamma\frac{\dot{\c}}{\left|\dot{\c}\right|}\qquad\textrm{if }\dot{\c}\neq 0,\\   
   &\mu-K{\c}\in [-\gamma,\gamma]\ \ \quad\qquad\textrm{if }\dot{\c}=0,
 \end{align}
\end{subequations}
and
\begin{equation}\label{eq:78}
  \frac{\dot{\c}}\tau=\alpha\Delta\mu.
\end{equation}
We are interested in the (formal) limit 
\begin{equation}
  \tau\to\infty,
\end{equation}
 which describes a regime when the boundary datum changes slowly. In this regime, \eqref{eq:77} and \eqref{eq:78} become, respectively,
\begin{subequations}\label{eq:79}
 \begin{align}
   &\mu=K{\c}
+\gamma\frac{\dot{\c}}{\left|\dot{\c}\right|}\ \ \qquad\qquad\textrm{if }\dot{\c}\neq 0,\\   
   &\mu-K{\c}\in [-\gamma,\gamma]\ \ \quad\qquad\textrm{if }\dot{\c}=0,
 \end{align}
\end{subequations}
and
\begin{equation}\label{eq:80}
  \alpha\Delta\mu=0.
\end{equation}
Now, in view of the boundary condition \eqref{eq:81}, it follows from \eqref{eq:80} that $\mu$ is spatially constant in $\Omega$ at each time, that is:
\begin{equation}
  \mu(x,s)=A z(s).
\end{equation}
Given that the initial condition for $u$ does not depend on $x$ as well, we can look for \emph{constant-in-space solutions}:
\begin{equation}
u(x,s)=v(s), 
\end{equation}
with $v(s)$ satisfying
\begin{subequations}\label{eq:82}
 \begin{align}
 &A z(s)=Kv(s)+\gamma \frac{\dot v(s)}{\left|\dot v(s)\right|}\ \ \hskip 0.1em \qquad\textrm{if }\dot v(s)\neq 0,\\
 &A z(s)-Kv(s)\in [-\gamma,\gamma]\qquad
\quad\textrm{if }\dot v(s)=0.
 \end{align}
\end{subequations}
Given that $z(0)=0$ by \eqref{eq:113}, system \eqref{eq:82} has a unique solution $v\in H^1(0,T;\mathbb R)$ satisfying the initial condition $v(0)=0$ \cite[Lemma 2.9]{Visin1994}. This solution, which in general is described by a \emph{play operator} $\mathcal E$ such that $v(\cdot)=\mathcal E(Az(\cdot),v(0))$, can be worked out explicitly in the present case. 

If the half-amplitude $A$ of the oscillation is smaller than or equal to the threshold $\gamma$, the solution is null at all times; otherwise, for $\mathcal P$ the projection on the closed interval $[-A+\gamma,A-\gamma]$, the solution is given by:
\begin{equation*}
Kv(s)=\left\{
\begin{array}{ll}
0&\textrm{if }s\in\left[0,\frac \gamma {4A}\right),\\[0.5em]
\mathcal P\left(Az\left(s-\frac \gamma {4A}\right)\right)&\textrm{if }s\in\left[\frac \gamma {4A},+\infty\right),
\end{array}\right.
\end{equation*}
as shown in Fig. 1.(a) below.
\begin{figure}[h]
\centering
\includegraphics[scale=0.45]{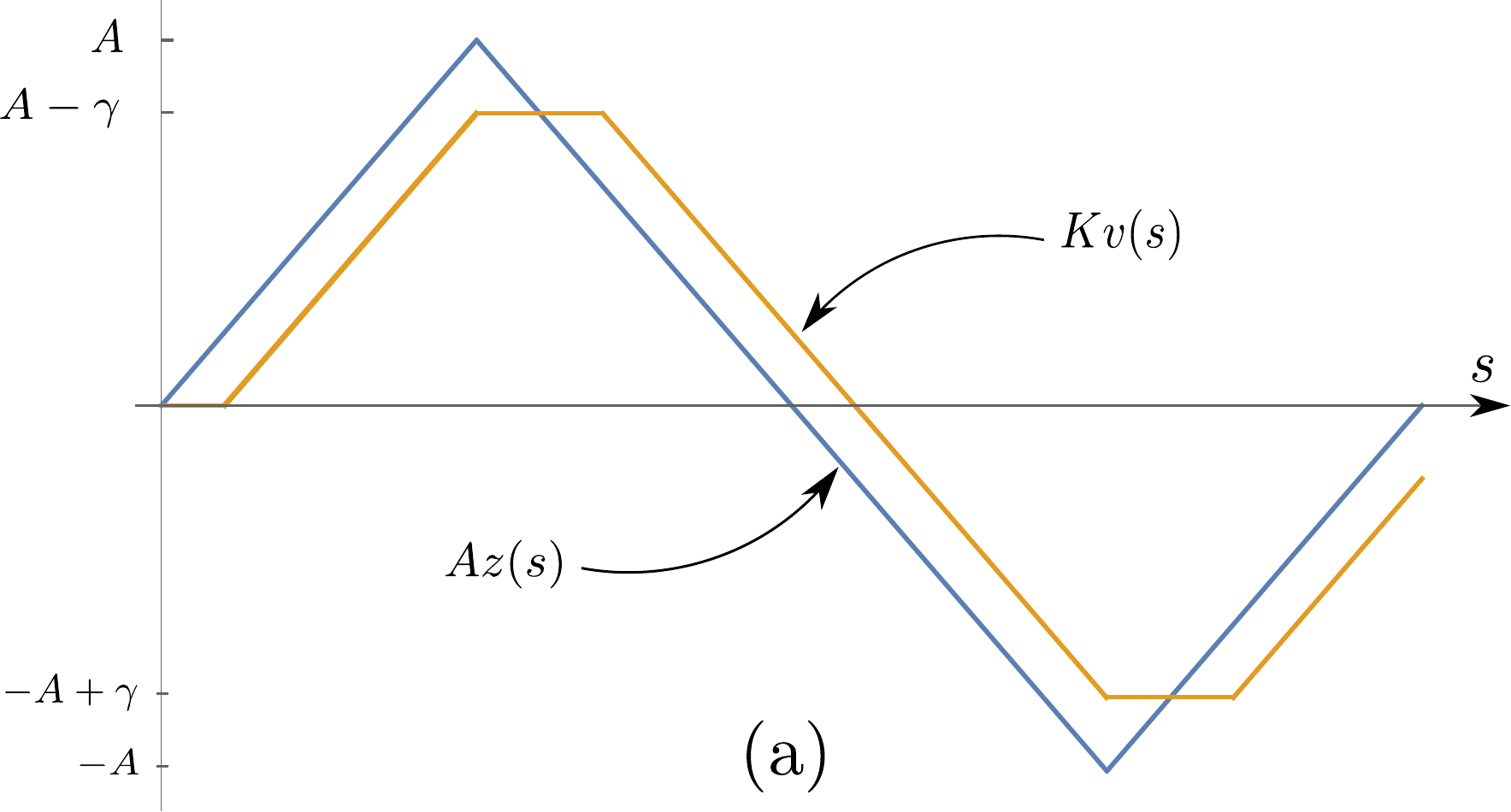}\quad\includegraphics[scale=0.45]{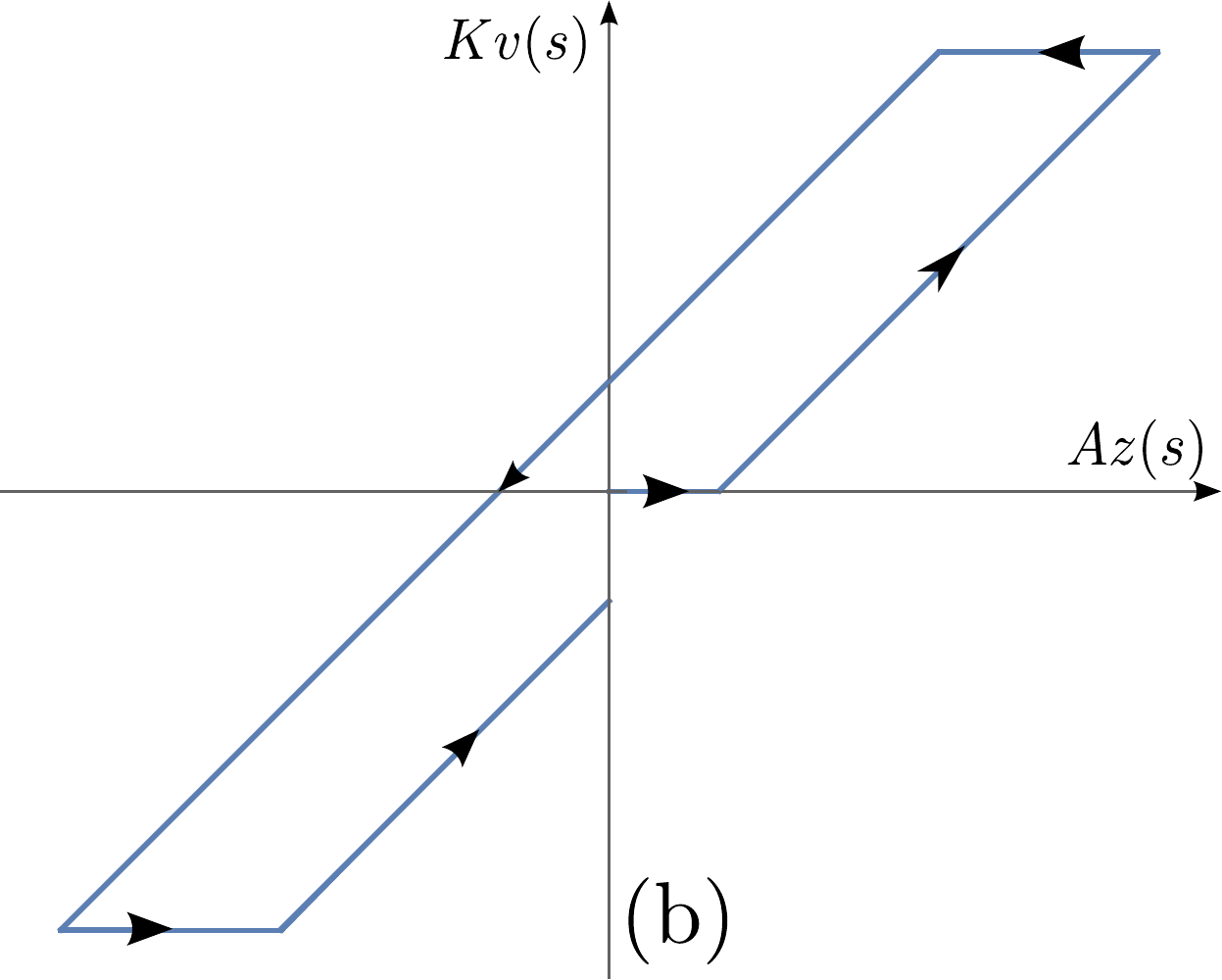}
\caption{(a) plots of $A z(s)$ and $K v(s)$; (b) parametric plot of $s\mapsto (A z(s),Kv(s))$. Here $s=t/\tau$ is the rescaled time.}
\end{figure}
The graphs of $Az(s)$ along with $K v(s)$ are shown in Fig. 2. In particular, hysteresis is apparent from Fig. 1.(b).
\bigskip

\textbf{Singular free energies.}
One may want to consider the following singular free energy suggested by the \emph{regular solution model} \cite{AtkinP2006}. In this model, the \emph{entropic contribution} 
\begin{equation}
  \psi_{\rm e}(\cc)=k r\log r,\qquad k>0,
\end{equation}
to the free energy is augmented by a quadratic \emph{enthalpic term} through an \emph{interaction parameter} $\chi>0$, so that the total free energy is:
\begin{equation}
  \label{eq:41}
  \psi({r})=\psi_{\rm e}({r})+\psi_{\rm e}(1-r)+\chi{r}(1-{r})\quad \textrm{if }r\in(0,1),
\end{equation}
with $\psi(r)=+\infty$ otherwise. This prototypical energy is, for instance, relevant when describing solvent diffusion in polymer gels (see, \emph{e.g.}, \cite{Doi2009JPSJ,Flory1942TJocp,lucantonio2013transient}).

Now, the issue with the singular terms in  \eqref{eq:41} is that, at variance with energies having polynomial growth,  $\psi_{\rm e}'(r)$ is not controlled by $\psi_{\rm e}(r)$. In fact, as $r\in (0,1)$ approaches $0$ or $1$, we have that $\psi_{\rm e}(r)$ tends to null, whereas $|\psi_{\rm e}'(r)|$ blows up, as shown in Fig. 2 below.
\begin{figure}[h]
\centering
\includegraphics[scale=2]{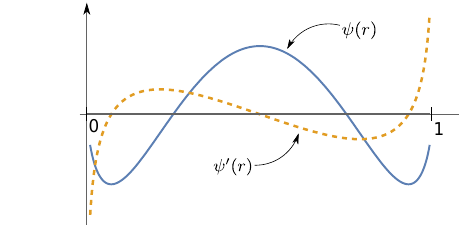}
\caption{Plots of the free energy in \eqref{eq:41} and of its derivative (respectively, solid and dashed line).}
\end{figure}
In this case, the argument leading from \eqref{eq:10} to the energetic estimates \eqref{eq:2211} breaks down, even when carried out formally. Indeed, the bound \eqref{eq:37}, which is essential to that argument, follows from (the discrete counterpart of) the equation
\begin{equation}\label{eq:94}
\mu=\beta\dot{\c}+\psi'(u),  
\end{equation}
whose exploitation is contingent on having a control on $\psi'(u)$. For energies with polynomial growth, such control is guaranteed from $\psi'(u)$ being bounded by $\psi(u)$, but this is not the case when $\psi(u)$ has the form \eqref{eq:41}.

The fact that a basic energetic estimate fails in this situation is the symptom of a structural feature that would be fatal to any na\"ive attempt to handle the problem by approximation of the free energy. Suppose in fact that, by mimicking \cite{EllioL1991MA}, we approximate $\psi_{\rm e}$ with the following regularized functions:
\begin{equation}
  \label{eq:58}
  \psi_{\rm e,\varepsilon}({r})=\begin{cases}
\psi_{\rm e}({r})\quad\textrm{if}\quad {r}\ge \varepsilon,
\\[0.5em]
\displaystyle kr\log\varepsilon+\frac k 2 \Big(\frac{r^2}{\varepsilon}-\varepsilon \Big)\quad\textrm{if}\quad{r}<\varepsilon,
\end{cases}
\end{equation}
and we consider a $\varepsilon$-parametrized sequence of problems with the following regularized free energy:
\begin{equation}\label{eq:33}
  \psi_\varepsilon({r})=\psi_{\rm e,\varepsilon}({r})+\psi_{\rm e,\varepsilon}(1-{r})+\chi {r}(1-{r}).
\end{equation}
The energies $\psi_\varepsilon$ specified by \eqref{eq:33} are smooth and have quadratic growth as $|r|\to+\infty$. Thus, Theorem \ref{thm:1} can be applied to establish the existence, for each particular $\varepsilon$, of an approximation
\begin{align}
  \label{eq:56}
&{\c}_\varepsilon\in L^\infty(0,T;H^1(\Omega))\cap H^1(0,T;L^2(\Omega)),
\\
&\mu_\varepsilon\in L^2(0,T;H^1(\Omega)),
\end{align}
that satisfies
\begin{subequations}\label{eq:96}
\begin{align}
&\int_\Omega \dot{\c}_\varepsilon(t) v+\alpha\nabla\mu_\varepsilon(t)\cdot\nabla v=\int_{\Gamma} hv\quad\forall v\in H^1(\Omega)\text{ for a.a. } t\in(0,T),\label{eq:59}\\
&\mu_\varepsilon=\alpha\dot{\c}_\varepsilon+\psi_\varepsilon'({\c}_\varepsilon)\quad \textrm{a.e. in } Q,\label{eq:522}\\
&{\c}_\varepsilon(0)={\c}_0 \quad \textrm{a.e. in }\Omega.\label{eq:522bis}
\end{align}
\end{subequations}
Since $\psi_\varepsilon$ approximates $\psi$ in the sense that
\begin{equation}\label{eq:65}
  \psi'_\varepsilon(r)\to\psi'(r)\quad\textrm{ as $\varepsilon\to 0$ for all }r\in(0,1),
\end{equation}
one may hope that, as $\varepsilon\to 0$, the solution of \eqref{eq:96} converge to a limit, and this limit satisfy the original problem. 

In order to establish such convergence, however, one should rely on estimates analogous to \eqref{eq:2211}, \eqref{eq:43}, and \eqref{eq:233}. The problem however is that the bounds in these estimates would not be uniform as $\varepsilon\to 0$, because  $\psi_\varepsilon(r)$ bounds $\psi'_\varepsilon(r)$ only through a multiplicative constant that blows up as $\varepsilon\to 0$. 

On the other hand, if the Neumann boundary condition \eqref{eq:93} were to be replaced by the Dirichlet condition:
\begin{equation}\label{eq:95}
 \mu=\mu_\flat\quad\textrm{on}\quad\Gamma,
\end{equation}
as done in \cite{BonetCT2015?}, then an energy balance similar to \eqref{eq:10} would hold. In this case, it would be possible to recover the estimates \eqref{eq:2211} on $\mu$ and $\dot{\c}$ without recourse to \eqref{eq:37}, using instead the control on $\mu$ provided by its trace. Eventually, one would be able to control $\psi'({\c})$ by comparison in \eqref{eq:94}. 

All in all, if the boundary condition \eqref{eq:95} is enforced, the estimates
\begin{align}
  &\|u_\varepsilon\|_{L^\infty(0,T;H^1(\Omega)}\le C,\label{eq:60}\\
  &\|\dot u_\varepsilon\|_{L^2(Q)}\le C,\\
  &\|\mu_\varepsilon\|_{L^2(0,T;H^1(\Omega))}\le C,\label{eq:61}
\end{align}
hold uniformly with respect to $\varepsilon$. Using these estimates, we can pass to the limit in \eqref{eq:522} by replicating the argument in \cite{EllioL1991MA}.\medskip

\textbf{Alternative frameworks.}
The point of view put forth in \cite{GurtiFA2010}, which we illustrate in Section 2, appears to be less general than \cite{FriedG1999JSP}, where a microforce balance:
\begin{equation}\label{eq:102}
  {\rm div}\bm\xi+\pi+\gamma=0
\end{equation}
augments the mass balance \eqref{eq:1} and chemical potential is included in the list of independent state variables. In particular, the \emph{internal microforce} $\pi$, rather than chemical potential, is the object of a constitutive prescription (see also the discussion in \cite{BonetCT2015?}). However, since we do not consider surface-tension effects, in our case the \emph{vectorial microstress} vanishes:
\begin{equation}
  \bm\xi=\mathbf 0,
\end{equation}
and, given that the \emph{external microforce} $\gamma$ vanishes in any actual evolution process:
\begin{equation}\label{eq:107}
  \gamma=0,
\end{equation}
the microscopic force balance \eqref{eq:102} reduces to the statement $\pi=0$, a statement that, when combined with the constitutive equation
\begin{equation}\label{eq:26}
  \pi=\mu-\widehat\psi'(u),
\end{equation}
yields the equation of state \eqref{eq:5}.

An alternative framework for  processes of phase diffusion, which hinges on the analogy between chemical potential and coldness as measures of \emph{orderliness}, has been proposed in \cite{Podio2006RM}. In that framework, the microforce balance \eqref{eq:102} is retained, but balance of microscopic energy and imbalance of microscopic entropy replace, respectively, mass balance and dissipation inequality. Consistent with this point of view, the microscopic entropy flux, which is involved in the microscopic entropy imbalance, is assumed to be proportional through $\mu$ to the  microscopic energy flux, which appears in the microscopic energy balance. With this assumption, the microscopic entropy imbalance can be combined with the microscopic balances of force and energy to arrive at the following counterpart of the conventional \emph{reduced dissipation inequality}:
\begin{equation}\label{eq:105}
\dot\psi\le -\eta(\mu^{-1})^{\displaystyle\huge\cdot}+\mu^{-1}\overline{\mathbf h}\cdot\nabla\mu-\pi\dot {\c}+\bm\xi\cdot\nabla\dot {\c},
\end{equation}
which involves, besides the aforementioned \emph{flux of microscopic energy} $\overline{\mathbf h}$, a \emph{microscopic free energy} $\psi$ and a \emph{microscopic entropy} $\eta$. In particular a collection of constitutive prescriptions are sorted out, leading to diffusion models \emph{\`a la Cahn-Hilliard} which have been studied in \cite{colli2011well,colli2013global}. These prescriptions include:
\begin{equation}\label{eq:103}
   \psi=\widetilde\psi(u,\nabla u,\mu),\qquad  \pi=-\partial_{\c}\widetilde\psi({\c},\nabla {\c},\mu),\qquad \bm\xi=\partial_{\nabla u}\widetilde\psi(u,\nabla u,\mu).
\end{equation}
In particular, if the constitutive mapping delivering the free energy density has the form:
\begin{equation}
  \widetilde\psi({\c},\nabla {\c},\mu)=-\mu {\c}+\widehat\psi(u)+\frac 12\kappa|\nabla u|^2,
\end{equation}
then the microforce balance \eqref{eq:102} with $\gamma=0$ leads to the equation of state \eqref{eq:101}. 

In this framework, the analogue of the viscous regularization \eqref{eq:53} would be: 
\begin{equation}\label{eq:104}
  \pi=-\partial_{\c}\widetilde\psi({\c},\nabla {\c},\mu)-\beta\dot{\c}.
\end{equation}
Moreover, the inclusion
\begin{equation}\label{eq:106}
   -\pi-\partial_{\c}\widetilde\psi({\c},\nabla{\c},\mu)\in\partial\widehat\zeta(\dot{\c}),
 \end{equation}
would take the place \eqref{eq:35}. The analytical consequences of both \eqref{eq:104} and \eqref{eq:106} are still to be explored.


\section*{Acknowledgements}
The author acknowledges the financial support of INdAM-GNFM through the initiative ``Progetto Giovani''. The author is grateful to Pierluigi Colli for his valuable 
feedback on an early draft of this paper. The author thanks also Eliot Fried, Amy Novick-Cohen, and Flavia Smarrazzo for providing relevant references.

\bibliographystyle{abbrv}
\bibliography{bibliography,bib2} 

\begin{thebibliography}{10}

\bibitem{AtkinP2006}
P.~Atkins and J.~de~Paula.
\newblock {\em Atkins' Physical Chemistry}.
\newblock W. H. Freeman and Company, 2006.

\bibitem{barenblatt1993degenerate}
G.~Barenblatt, M.~Bertsch, R.~D. Passo, and M.~Ughi.
\newblock A degenerate pseudoparabolic regularization of a nonlinear
  forward-backward heat equation arising in the theory of heat and mass
  exchange in stably stratified turbulent shear flow.
\newblock {\em SIAM J. Math. Anal.}, 24:1414--1439, 1993.

\bibitem{BertsPV2001AMPA}
M.~Bertsch, P.~Podio-Guidugli, and V.~Valente.
\newblock On the dynamics of deformable ferromagnets. {I}. {G}lobal weak
  solutions for soft ferromagnets at rest.
\newblock {\em Ann. Mat. Pur. Appl.}, 179:331--360, 2001.

\bibitem{bertsch2013pseudoparabolic}
M.~Bertsch, F.~Smarrazzo, and A.~Tesei.
\newblock Pseudoparabolic regularization of forward-backward parabolic
  equations: A logarithmic nonlinearity.
\newblock {\em Analysis \& PDE}, 6:1719--1754, 2013.

\bibitem{BonetCT2015?}
E.~Bonetti, P.~Colli, and G.~Tomassetti.
\newblock In preparation.
\newblock 2015.

\bibitem{CahnH1958JCP}
J.~W. Cahn and J.~E. Hilliard.
\newblock Free energy of a nonuniform system. {I}. {I}nterfacial free energy.
\newblock {\em J. Chem. Phys.}, 28:258--267, 1958.

\bibitem{ColemN1963ARMA}
B.~D. Coleman and W.~Noll.
\newblock {The thermodynamics of elastic materials with heat conduction and
  viscosity}.
\newblock {\em Arch. Ration. Mech. Anal.}, 13:167--178, 1963.

\bibitem{colli2011well}
P.~Colli, G.~Gilardi, P.~Podio-Guidugli, and J.~Sprekels.
\newblock Well-posedness and long-time behavior for a nonstandard viscous
  {C}ahn-{H}illiard system.
\newblock {\em SIAM J. Appl. Math.}, 71:1849--1870, 2011.

\bibitem{colli2013global}
P.~Colli, G.~Gilardi, P.~Podio-Guidugli, and J.~Sprekels.
\newblock Global existence and uniqueness for a singular/degenerate
  {C}ahn--{H}illiard system with viscosity.
\newblock {\em J. Diff. Eq.}, 254:4217--4244, 2013.

\bibitem{Doi2009JPSJ}
M.~Doi.
\newblock Gel dynamics.
\newblock {\em J. Phys. Soc. Jpn.}, 78:052001, 2009.

\bibitem{EllioL1991MA}
C.~Elliot and S.~Luckhaus.
\newblock Generalized diffusion equation for phase separation of a
  multi-component mixture with interfacial free energy.
\newblock {\em \rm Preprint 887, Institute for Mathematics and its
  Applications, Minneapolis}, 1991.

\bibitem{EllioS1996JDE}
C.~Elliott and A.~Stuart.
\newblock Viscous {C}ahn--{H}illiard {E}quation {II}. {A}nalysis.
\newblock {\em J. Diff. Eq.}, 128:387--414, 1996.

\bibitem{elliott1996cahn}
C.~M. Elliott and H.~Garcke.
\newblock On the {C}ahn-{H}illiard equation with degenerate mobility.
\newblock {\em SIAM J. Math. Anal.}, 27:404--423, 1996.

\bibitem{EvansP2004MMMAS}
L.~C. Evans and M.~Portilheiro.
\newblock Irreversibility and hysteresis for a forward-backward diffusion
  equation.
\newblock {\em Math. Mod. Meth. Appl. Sci.}, 14:1599--1620, 2004.

\bibitem{Flory1942TJocp}
P.~J. Flory.
\newblock Thermodynamics of high polymer solutions.
\newblock {\em J. Chem. Phys.}, 10:51, 1942.

\bibitem{FriedG1999JSP}
E.~Fried and M.~Gurtin.
\newblock Coherent solid-state phase transitions with atomic diffusion: A
  thermomechanical treatment.
\newblock {\em J. Stat. Phys.}, 95:1361--1427, 1999.

\bibitem{Gurti1996PD}
M.~E. Gurtin.
\newblock Generalized {G}inzburg-{L}andau and {C}ahn-{H}illiard equations based
  on a microforce balance.
\newblock {\em Phys. D}, 92:178--192, 1996.

\bibitem{GurtiFA2010}
M.~E. Gurtin, E.~Fried, and L.~Anand.
\newblock {\em The {M}echanics and {T}hermodynamics of {C}ontinua}.
\newblock Cambridge University Press, 2010.

\bibitem{Lions1969}
J.-L. Lions.
\newblock {\em {Quelques m\'{e}thodes de r\'{e}solution des probl\`{e}mes aux
  limites non lin\'{e}aires}}.
\newblock Dunod, 1969.

\bibitem{lucantonio2013transient}
A.~Lucantonio, P.~Nardinocchi, and L.~Teresi.
\newblock Transient analysis of swelling-induced large deformations in polymer
  gels.
\newblock {\em J. Mech. Phys. Solids}, 61:205--218, 2013.

\bibitem{miranville1999model}
A.~Miranville.
\newblock A model of {C}ahn--{H}illiard equation based on a microforce balance.
\newblock {\em Compt. Rend. Acad. Sci. - Ser. I-Math.}, 328:1247--1252, 1999.

\bibitem{miranville2000some}
A.~Miranville.
\newblock Some generalizations of the {C}ahn--{H}illiard equation.
\newblock {\em Asympt. Anal.}, 22:235--259, 2000.

\bibitem{miranville1998dynamical}
A.~Miranville, A.~Pietrus, and J.-M. Rakotoson.
\newblock Dynamical aspect of a generalized {C}ahn--{H}illiard equation based
  on a microforce balance.
\newblock {\em Asympt. Anal.}, 16:315--345, 1998.

\bibitem{Novic1988viscous}
A.~Novick-Cohen.
\newblock On the viscous {C}ahn-{H}illiard equation.
\newblock In {\em Material instabilities in continuum mechanics ({E}dinburgh,
  1985--1986)}, pages 329--342. Oxford University Press, 1988.

\bibitem{NovicP1991TAMS}
A.~Novick-Cohen and R.~L. Pego.
\newblock Stable patterns in a viscous diffusion equation.
\newblock {\em Trans. Amer. Math. Soc.}, 324:331--351, 1991.

\bibitem{Plotn1994Passing}
P.~I. Plotnikov.
\newblock Passing to the limit with respect to viscosity in an equation with
  variable parabolicity direction.
\newblock {\em Diff. Eq.}, 30:614--622, 1994.

\bibitem{Podio2006RM}
P.~Podio-Guidugli.
\newblock Models of phase segregation and diffusion of atomic species on a
  lattice.
\newblock {\em Ric. Mat.}, 55:105--118, 2006.

\bibitem{porzio2013radon}
M.~M. Porzio, F.~Smarrazzo, and A.~Tesei.
\newblock Radon measure-valued solutions for a class of quasilinear parabolic
  equations.
\newblock {\em Arch. Rat. Mech. Anal.}, 210:713--772, 2013.

\bibitem{Roubi2013Nonlinear}
T.~Roub{\'\i}{\v{c}}ek.
\newblock {\em Nonlinear partial differential equations with applications.}
\newblock Birkh\"auser, second edition, 2013.

\bibitem{Simon1987AMPA4}
J.~Simon.
\newblock Compact sets in the space {$L^p(0,T;B)$}.
\newblock {\em Ann. Mat. Pura Appl.}, 146:65--96, 1987.

\bibitem{thanh2014passage}
B.~L.~T. Thanh, F.~Smarrazzo, and A.~Tesei.
\newblock Passage to the limit over small parameters in the viscous
  {C}ahn--{H}illiard equation.
\newblock {\em J. Math. Anal. Appl.}, 420:1265--1300, 2014.

\bibitem{thanh2014sobolev}
B.~L.~T. Thanh, F.~Smarrazzo, and A.~Tesei.
\newblock Sobolev regularization of a class of forward--backward parabolic
  equations.
\newblock {\em J. Diff. Eq.}, 257:1403--1456, 2014.

\bibitem{victor1998sobolev}
P.~Victor.
\newblock Sobolev regularization of a nonlinear ill-posed parabolic problem as
  a model for aggregating populations.
\newblock {\em Comm. Part. Diff. Eq.}, 23:457--486, 1998.

\bibitem{Visin1994}
A.~Visintin.
\newblock {\em Differential models of hysteresis}.
\newblock Springer Berlin, 1994.

\bibitem{visintin2002forward}
A.~Visintin.
\newblock Forward--backward parabolic equations and hysteresis.
\newblock {\em Calc. Var. Partial Diff. Eq.}, 15:115--132, 2002.

\end{thebibliography}

\end{document}